\documentclass[conference,10pt]{IEEEtran}

\usepackage[utf8]{inputenc} 
\usepackage[T1]{fontenc}    
\usepackage{times}          
\usepackage{graphicx}       
\usepackage{amsmath, amssymb, amsthm} %
\usepackage{amsfonts}
\usepackage{xcolor}

\usepackage{cite}           
\usepackage[    
            letterpaper,
            top=0.7in,
            bottom=1.1in,
            left=0.65in,
            right=0.65in
           ]{geometry}

\usepackage{natbib}

\newtheorem{theorem}{Theorem}
\newtheorem{lemma}{Lemma}
\newtheorem{corollary}{Corollary}
\newtheorem{remark}{Remark}
\newtheorem{proposition}{Proposition}
\title{A DPI‑PAC‑Bayesian Framework for Generalization Bounds}

\author{
\IEEEauthorblockN{Muhan Guan\IEEEauthorrefmark{1},
                  Farhad Farokhi\IEEEauthorrefmark{1},
                  Jingge Zhu\IEEEauthorrefmark{1}}
\IEEEauthorblockA{\IEEEauthorrefmark{1}%
Department of EEE, University of Melbourne, Parkville, Victoria, Australia\\
Email: muhang@student.unimelb.edu.au,\;
       \{farhad.farokhi, jingge.zhu\}@unimelb.edu.au}
}
\begin{document}

\maketitle
\thispagestyle{empty}

\begin{abstract}
  We develop a unified Data Processing Inequality PAC-Bayesian framework—abbreviated DPI-PAC-Bayesian—for deriving the generalization error bounds in the supervised learning setting. By embedding the Data Processing Inequality (DPI) into the change‑of‑measure technique, we obtain explicit bounds on the binary Kullback–Leibler generalization gap for both Rényi divergence and any $f$-divergence measured between a data‑independent prior distribution and an algorithm‑dependent posterior distribution. We present three bounds derived under our framework using Rényi, Hellinger \(p\) and Chi-Squared divergences. Additionally, our framework also demonstrates a close connection with other well-known bounds. When the prior distribution is chosen to be uniform, our bounds recover to the classical Occam’s Razor bound and, crucially, eliminate the extraneous \(\log(2\sqrt{n})/n\) slack present in the PAC‑Bayes bound, thereby achieving tighter bounds. The framework thus bridges data‑processing and PAC‑Bayesian perspectives, providing a flexible, information‑theoretic tool to construct generalization guarantees.\end{abstract}

\begin{IEEEkeywords}
  DPI‑PAC‑Bayesian framework, generalization bounds, Data Processing Inequality, PAC-Bayes bound
\end{IEEEkeywords}


\section{Introduction}
\label{intro}
Bounding techniques under supervised learning setting can provide theoretical guarantees for the performance of machine learning models on unseen data, improving the generalization capabilities of the models. 

This work focuses on high-probability generalization bounds. A typical result states that, with probability at least $1-\delta$ over a set of \textit{i.i.d.} samples, the population risk of a model is upper-bounded by
$f\bigl(\delta,\;\text{empirical risk on the samples}\bigr)$. By contrast, information-theoretic bounds usually bound the expected gap between population and empirical risks. A thorough comparison of these two families of bounds is provided in \cite{MAL-112}.

As the work of Langford \cite{JMLR:v6:langford05a} mentioned, the high-probability generalization bounds for supervised learning setting can be classified into two classes: test-set bounds and train-set bounds. 
Both of these bounds have their advantages and disadvantages.
Although test-set bounds can give a tight upper bound on the error rate on unseen data, the main problem of such bounds is that the data used to evaluate the bounds cannot be used for learning. Specifically, we have to remove some training examples and keep them as a holdout set, which could lead to loss of performance on our learned hypothesis when training examples are inadequate. 

Compared to test-set bounds, train-set bounds are the current focus of learning theory work.
The biggest advantage of train-set bounds is that we can use entire data samples to perform both learning and bound construction, but many train-set bounds are generally loose. Therefore, it is crucial to develop techniques to improve the tightness of train-set bounds, so that these bounds can provide better insight into the learning problem itself.
\subsection{Our contribution}
In this work, we propose a flexible DPI-PAC-Bayesian framework for deriving train-set generalization error bounds under the supervised learning setting by combining Data Processing Inequality (DPI) with the spirit of the PAC-Bayesian perspective. This framework accommodates Rényi divergence and also arbitrary $f$-divergence measures.

In addition to its flexibility, the framework shows a close connection to other widely used train-set bounds and also yields provably tight bounds. Our theoretical results demonstrate that, in some special cases, the bounds derived by our framework can recover to the Occam's Razor bound and also can be explicitly tighter than the PAC-Bayes bound.

\subsection{Problem setting}
Consider a standard supervised learning setting. We have $n$ \textit{i.i.d.} training samples $S=\{Z_1,...Z_n
\}$, which are randomly drawn from an underlying data-generating distribution $\mathcal{D}$. A hypothesis space $\mathcal{W}$ that includes a set of hypotheses (or classifiers) $w$. The learning algorithm is treated as a conditional probability distribution \( P_{W|S} \). For a given training set \( s \), the algorithm samples a hypothesis \( w \) according to \( P_{W|S=s} \). Coupled with a marginal distribution \( P_S \) over training samples, this defines a joint distribution over hypothesis space and data, given by \( P_{S,W} = P_S P_{W|S} \) on \( \mathcal{W} \times \mathcal{Z}^n \).  The performance of a hypothesis $w\in\mathcal{W}$ on a training sample is measured by a loss function $\ell:\mathcal{W}\times\mathcal{Z}\rightarrow[0,1]$. The empirical loss of $w$ is defined as $\hat{L}(S,w)=\frac{1}{n}\sum_{i=1}^n\ell(Z_i,w)$, while the population loss of $w$ on $\mathcal{D}$ is $L(w)=\mathbb{E}_{Z\sim\mathcal{D}}\{\ell(Z,w)\}$. For any $w\in\mathcal{W}$, we consider the bounds on the generalization gap $L(w)-\hat{L}(S,w)$.
For ease of exposition, throughout this work we consider on the \emph{finite‑hypothesis} case $|\mathcal{W}|<\infty$. The prior distribution $Q_W$ assigns a strictly positive mass to every $w\in\mathcal{W}$ and $\min_wQ_W(w)>0$, but the same technique can be extended to more general case when the hypothesis space $\mathcal{W}$ is infinite.

Consider a fixed kernel $W(y|x)$ and two different probability distributions $P_X$ and $Q_X$ defined on the same space $\mathcal{X}$. Then, define $P_Y(y) = \sum_x W(y|x) P_X(x)$ \text{and} $Q_Y(y) = \sum_x W(y|x) Q_X(x)$. 
Moreover, for a convex function \(f:(0,+\infty) \to \mathbb{R}\) satisfying \(f(1) = 0\), the \(f\)-divergence between two distributions on a probability space $\mathcal{X}$ is defined as
\[
D_f(P_X\|Q_X) = \sum_{x \in \mathcal{X}} Q_X(x)\, f\!\left(\frac{P_X(x)}{Q_X(x)}\right).
\]

We introduce the expressions of Rényi divergence and two crucial $f$-divergences for our bounds derivation:
1. Rényi divergence ($\alpha>0$, $\alpha\neq1$)
\[D_\alpha(P\|Q)=\frac{1}{\alpha-1}\ln(\sum_{x}P_X(x)^{\alpha}Q_X(x)^{1-\alpha}).
\]
2. Pearson $\chi^2$-divergence
\[\chi^2(P||Q)=\sum\limits_x\frac{P_X(x)^2}{Q_X(x)}-1.\]
3. Hellinger $p$-divergence ($p>0$, $p\neq1$)
\[\mathcal{H}^p(P||Q)=\frac{\sum\limits_xP_X(x)^pQ_X(x)^{1-p}-1}{p-1}.\]

\begin{proposition}[Data Processing Inequality \cite{Gastpar22}]
\label{prop:DPI}
With the distributions \(P_X,Q_X,P_Y,Q_Y\) and the kernel $W(y|x)$ defined previously, we have
\begin{align*}
&\text{(i) } D_f(P_Y\|Q_Y)\le D_f(P_X\|Q_X),\\
&\text{(ii) } D_\alpha(P_Y\|Q_Y)\le D_\alpha(P_X\|Q_X).
\end{align*}
\end{proposition}
That is, passing $P_X$ and $Q_X$ through the same kernel will make them ``more similar''. 

\section{Related work}
\label{RelatedWork}
To measure the discrepancy between empirical and expected losses, we employ the Kullback-Leibler (KL) function, $\mathrm{KL}(\hat{L}(S,W)||L(W))$. 
For $p, q\in[0,1]$, the KL function is defined as
\[
\mathrm{KL}(p||q)=p\ln\frac{p}{q}+(1-p)\ln\frac{1-p}{1-q}.
\]
The KL-loss form bounds can be further relaxed using Pinsker's inequality $\mathrm{KL}(p||q)\geq2(p-q)^2$ and then yield the bounds of the classical difference-loss form. 

In the supervised setting, generalization bounds fall into two categories: test‑set bounds, which require that an extra subset of the data be held out solely for evaluation, because these examples cannot be used during training, and train-set bounds, which use the entire dataset both to learn the hypothesis and to compute the bound.

\subsection{Test-set bound}
To illustrate how a test-set bound is evaluated, consider the following two-party scenario \cite{JMLR:v6:langford05a}: (1) A learner trains a hypothesis $w$ on a data set that the verifier will never see, and then transmits this fixed $w$ to the Verifier. (2) A verifier samples a set of data $S$ , using $w$ together with the empirical loss on $S$, computes the right-hand side of the bound.

In test-set bound, $S$ is generated after $w$ is fixed, and is independent of the learner’s training data.

\begin{theorem}
(KL test-set bound \cite{JMLR:v6:langford05a}).

For any fixed $w$, with probability at least $1-\delta$ over $P_S$, it holds that
\begin{equation}
 \mathrm{KL}( \hat{L}(S, w)||L(w))  \le \frac{\log \frac{1}{\delta}}{n} .
\label{eq:KL-test-bound}
\end{equation}
\end{theorem}
The bound is very simple and can be seen as computing a confidence interval for the binomial distribution as in \cite{clopper1934fiducial}.
\subsection{Train-set bound}
In a train-set bound, the same set of $S$ is used twice—first to train the hypothesis and then to evaluate the bound. The evaluation protocol is as follows:  
(1) A learner chooses a prior $Q_W(w)$ over the hypothesis space before seeing $S$ and sends it to the verifier. (2) The verifier samples the data $S$ and sends it to the learner. (3) The learner chooses $w$ based on $S$ and sends it to the verifier. (4) The verifier evaluates the bound.

The first and tightest train-set bound is the Occam’s Razor bound \cite{blumer1987occam}, and Langford \cite{JMLR:v6:langford05a} has proved that the bound cannot be improved without incorporating extra information.
\begin{theorem}
\label{or_bound_theorem}
(Occam's Razor bound \cite{blumer1987occam}).
\textit{Assume} \( w \in \mathcal{W} \) \textit{with} \( \mathcal{W} \) \textit{a countable set. Let} \( Q_{W}(w) \) \textit{be a distribution over} \( \mathcal{W} \). \textit{For any} $\delta\in(0,1]$ \textit{, with probability at least} $1-\delta$ \textit{over} $P_S$ \textit{, it holds that}
\begin{flalign}
\forall w,\,
  \mathrm{KL}(\hat{L}(S,w)||L(w))
  \le \frac{\log\frac1{Q_{ W}(w)}+\log\frac1\delta}{n}.
\label{or_bound}
\end{flalign}
\end{theorem}
PAC-Bayes bounds \cite{mcallester1998some} are also train-set bounds. We present the PAC-Bayes-KL bound from the work of McAllester~\cite{10.1007/978-3-540-45167-9_16}, which is one of the tightest known PAC-Bayes bounds in the literature and can be relaxed in various ways to obtain other PAC-Bayes Bounds~\cite{NIPS2013_a97da629}.
\begin{theorem}
    (PAC-Bayes bound \cite{10.1145/307400.307435}). \textit{Let} \( P_{W|S} \) \textit{be a fixed conditional distribution (given data $S$}) \textit{on} \( \mathcal{W} \). \textit{Define} $P[L] := \mathbb{E}_{P_{W|S}} \left\{ L(W) \right\}, \ P[\hat{L}] := \mathbb{E}_{P_{W|S}} \left\{ \hat{L}(S,W) \right\}. $ \textit{For any} $\delta\in(0,1]$ \textit{, with probability at least} $1-\delta$ \textit{over} $P_S$ \textit{, it holds that}
\begin{align}
 \mathrm{KL}(P[\hat{L}]||P[L]) \leq \frac{D(P_{W|S} \,\|\, Q_ W) + \log \frac{2\sqrt{n}}{\delta}}{n},
\label{eq:OR}
\end{align}

\label{theorem2}
\end{theorem}
\noindent
where $Q_{W}(w)$ is a prior distribution over the hypothesis space $\mathcal{W}$--specified before seeing training samples $S$. Furthermore, in the PAC-Bayesian framework, $W\sim P_{{W}|S}$ is the output of the posterior distribution (or the learning algorithm) on $S$.

\subsection{Comparing Occam's Razor (OR) and PAC-Bayes Bound}
To compare the two bounds, we specialize the PAC‑Bayes bound \eqref{eq:OR} by choosing the posterior $P_{W\mid S}(w')=\mathbf{1}_{\{w'=w\}}$. Then the term $\mathrm{KL}(P_{W|S}||Q_ W)=\log(1/Q_ W(w))$ The PAC-Bayes bound becomes
\begin{align}
\forall w,\  \,\mathrm{KL}(\hat{L}(S, w)||L(w))   \leq \frac{\log \frac{1}{Q_{W}(w)}+\log \frac{2\sqrt{n}}{\delta}}{n}.
\label{pac-bayes-pointmass}
\end{align}

Comparing this to the OR bound \eqref{eq:OR}, we see an extra term $\log 2\sqrt{n}$.
Then an open question that is worth studying \cite{langford2002quantitatively}: \textbf{Can the PAC-Bayes bounds be as tight as the OR bound?} To be specific, when specializing $P_{{W}|S}$ to a deterministic algorithm, can we remove the term $\log2\sqrt{n}$ from the PAC-Bayes bounds? 
Among the few works on this problem, \cite{10.1561/2200000100} investigates more general PAC‑Bayes bounds that use alternative $d$ functions:
\begin{align*}
\mathbb{P}_{S}\Bigl\{
\forall P_{W},\ 
\lambda\,d\bigl(P[L],\,P[\hat L]\bigr)
&\le
D\bigl(P_{W}\,\Vert\,Q_{W}\bigr)
+ \log\Phi(\lambda)\\
&\quad\;+\;\log\frac1\delta
\Bigr\}
\;\ge\;
1 - \delta.
\end{align*}

For example, if we use Catoni's function $C_\beta$ \cite{catoni2007pac}, and optimize the parameter $\beta$, then the $\log 2\sqrt{n}$ term (from $\Phi(\lambda)$) will be removed, which is not allowed for a classic PAC-Bayes bound \cite{foong2021tight}. Here we study this problem from different perspectives. Our work aims to bridge the gap between the PAC-Bayesian framework and the OR bound by the DPI.

\section{Main results}

\subsection{Some useful lemmas}
Following the same technique introduced in~\cite{9444402}, we derived three key lemmas by combining DPI with the Rényi divergence and two $f$-divergences. In particular, a similar result for the KL divergence appeared in \cite{Gastpar22}.

In this subsection, we define two probability spaces \( ({\Omega}, \mathcal{F},P) \), \( ({\Omega}, \mathcal{F},Q) \), where ${\Omega}=\mathcal{X\times Y}$. Let $E\in\mathcal{F}$ be a (measurable) event.
\begin{lemma}(Change of measure with Rényi Divergence).
\label{DPI_original_bound_1}
For any \( \alpha > 1 \) and any event \( E \in \mathcal{F} \), we have the bound
\begin{equation}
P(E)\le Q(E)^{\frac{\alpha-1}{\alpha}}e^{\frac{\alpha-1}{\alpha}D_\alpha(P||Q)}.
\end{equation}
\end{lemma}
\begin{proof}
We define a fixed kernel $W(y|x)$ to generate $P_Y(y)$ and $Q_Y(y)$:
\[
\begin{cases} 
     W(y=1|x) = 1,\ \textit{if} \ x\in E,   \\ 
     W(y=0|x) = 1,\ otherwise. 
\end{cases}
\]
Notice for all $x$, we have $W(y=1|x)+W(y=0|x)=1$.
By Proposition~\ref{prop:DPI}, the Rényi divergence satisfies
\[
D_\alpha(P_X \| Q_X) \ge D_\alpha(P_Y \| Q_Y), \quad \text{for any } \alpha \in (1, \infty).
\]
Since \(Y \in \{0,1\}\), the RHS becomes
\begin{align*}
D_\alpha \big( P_Y(y) \| Q_Y(y) \big) = \frac{1}{\alpha -1} \log \sum_{y \in \{0,1\}}P_Y(y)^\alpha Q_Y(y)^{1-\alpha},
\end{align*}
where\begin{align*}
P_Y(y=1)&= \sum_{x \in E} W(y=1|x) P_X(x) \\
&\quad\ + \sum_{x \notin E} W(y=1|x) P_X(x)\\ &= P(E),
\end{align*}
and $P_Y(y=0)= 1 - P(E)$. Similarly, we have $Q_Y(y=1) = Q(E)$, $Q_Y(y=0) = 1 - Q(E)$. Thus when $\alpha\in(1,\infty)$, we have
\begin{align*}
D_{\alpha}(P_X(x)||Q_X(x)) &\geq \frac{1}{\alpha-1}\log [P(E)^\alpha Q(E)^{1-\alpha}\\
&\phantom{\geq\frac{1}{\alpha-1}} +(1 - P(E))^\alpha(1 - Q(E))^{1-\alpha}]\\
&\geq \frac{1}{\alpha-1}\log [P(E)^\alpha Q(E)^{1-\alpha}],
\end{align*}
then we have proved the Lemma~\ref{DPI_original_bound_1} by rearranging the above inequality.
\end{proof}
\begin{lemma}(Change of measure with Hellinger $p$-Divergence). 
\label{DPI_original_bound_hellinger_p_lemma}
 For any $p>1$ and any event \( E \in \mathcal{F} \) such that \( P(E) < \frac{1}{2} \) and \( Q(E) < \frac{1}{2} \), we have the bound
\begin{align}
P(E)\le \left[1+Q(E)^{(1-p)}\right]^{-\frac{1}{p}}\left[(p-1)\mathcal{H}^{p}(P||Q)+1\right]^{\frac{1}{p}}
\label{DPI_original_bound_hellinger_p_inequality}.
\end{align}
\end{lemma}
\begin{proof}
We define the same fixed kernel $W(y|x)$ as in the proof of Lemma~\ref{DPI_original_bound_1}. Thus for Hellinger $p$-Divergence we have
\begin{align*}
\mathcal{H}^{p}(P_X(x)\|Q_X(x)) 
&\geq \frac{1}{p-1} \left[
(1 - P(E))^p (1 - Q(E))^{1 - p} \right. \notag\\
&\quad\quad\quad\quad\quad\left. + (P(E))^p (Q(E))^{1-p} - 1 \right].
\end{align*}

\noindent
Using the assumptions $P(E)<\frac{1}{2}$ and $Q(E)<\frac{1}{2}$, we can further lower-bound the RHS as
\[\frac{1}{p-1} \left[P(E)^p(1+Q(E)^{(1-p)}) - 1 \right].\]
The claimed result follows by rearranging the terms.
\end{proof}
The conditions $P(E)<\frac{1}{2}$ and $Q(E)<\frac{1}{2}$ are naturally satisfied when \( E \) is defined as the failure event in which the KL-based test-set bound does not hold.

\begin{lemma}(Change of measure with Pearson Chi-Squared Divergence). 
\label{DPI_original_bound_Chi_Squared_lemma}
For any event \( E \in \mathcal{F} \), we have the bound
\begin{align}
P(E)\le Q(E)^\frac{1}{2}(\chi^2(P||Q)+1)^\frac{1}{2}.
\label{DPI_original_bound_Chi_Squared}
\end{align}
\end{lemma}
\begin{proof}
See Appendix~\ref{proof_lemma}
\end{proof}
These three lemmas are inspired by the change-of-measure principle commonly employed in the PAC-Bayesian framework. In PAC-Bayes analysis, the Donsker-Varadhan inequality enables one to bound expectations under an intractable posterior by reweighting expectations under a tractable prior distribution, typically introducing a KL divergence term to quantify the complexity of the posterior.

Within our framework, we exploit the DPI to upper‑bound the posterior distribution \(P(E)\) through several \(f\)-divergences between \(P\) and \(Q\); this idea was first introduced in \cite{9444402}. Each bound comprises a scaled prior term, such as $Q(E)^\gamma$, multiplied by an exponential penalty term that depends on the chosen divergence. These multiplicative correction factors—e.g., $e^{\frac{\alpha - 1}{\alpha} D_\alpha(P \| Q)}$ in the Rényi case—can be viewed as the cost of performing a change of measure from $Q$ to $P$ under the respective divergence. 

This perspective highlights a unifying theme across our results: DPI provides an information-theoretic control analogous to that in PAC-Bayes, enabling generalization bounds through divergence-based reweighting of prior knowledge.


\subsection{DPI-PAC-Bayes bounds}
Building on the preceding lemmas that fuse the DPI with Rényi ($D_\alpha$), Chi-Squared ($\mathcal{\chi}^2$), and Hellinger ($\mathcal{H}^p$) $p$-divergences, we now establish the core results developed by our DPI‑PAC‑Bayesian framework. The next three theorems present train‑set generalization bounds in terms of \(D_\alpha\), \(\mathcal{H}^p\), and \(\chi^2\) divergences developed by our framework.
\begin{theorem}
    ($D_\alpha$-PAC-Bayes bound). Let $Q$ be a distribution over a finite hypothesis space $\mathcal{W}$ such that $Q_{\min} := \min_w Q_{W}(w)>0$. For any $\alpha>1$, and for any $\delta\in(0,1]$, with probability at least $1-\delta$ over $P_S$, it holds that
\begin{align}
\forall w,\ 
\mathrm{KL}(\hat{L}(S, w)||L(w))
\leq 
\frac{
\log \frac{1}{Q_{min}} 
+ \frac{\alpha}{\alpha - 1} \log \frac{1}{\delta}
}{
n}.
\end{align}
\label{PAC-Bayes_with_D_alpha}

\end{theorem}

\begin{proof}
We choose \( P_{S,{W}} = P_S P_{{W}|S} \), and \( Q_{S,{W}} = P_S Q_{W} \) for some \( Q_{W} \), and define the event
\begin{align}
\label{Event_general}
E := \left\{ (S, W) : \mathrm{KL}( \hat{L}(S, W)||L(W)) \geq \frac{\log \frac{1}{\delta}}{n} \right\},
\end{align}

For any specific $w$, we define the event
\begin{align}
\label{Event_for_specific_w}
  E_w := \left\{ S: \mathrm{KL}( \hat{L}(S, w)||L(w)) \geq \frac{\log \frac{1}{\delta}}{n} \right\}.  
\end{align}

Then we apply the Lemma~\ref{DPI_original_bound_1} to get

\[
P(E) \leq Q(E)^{\frac{\alpha -1}{\alpha}} \left(\sum_{w,s} P(s,w)^\alpha Q(s,w)^{1-\alpha} \right)^{\frac{1}{\alpha}}.
\]
When $\alpha>1$, for $ Q(E)^{\frac{\alpha -1}{\alpha}}$ we have 
\begin{align*}
    Q(E)^{\frac{\alpha -1}{\alpha}}&= \left( \sum_{w} Q_W(w) \, \mathbb{P}\{E_w\} \right)^{\frac{\alpha - 1}{\alpha}} \\
&\le \left( \delta \sum_{w} Q_W(w) \right)^{\frac{\alpha - 1}{\alpha}}\\
&= \delta^{\frac{\alpha - 1}{\alpha}},
\end{align*}
where we used $\sum_{w} Q_W(w)=1$, and also the test-set bound in Theorem~\ref{eq:KL-test-bound} which states $\mathbb{P}\{E_w\}\leq\delta$. Furthermore, we have
\begin{align*}
\sum_{w,s} P(s,w)^\alpha Q(s,w)^{1-\alpha}
&= \sum_{w,s} P_S(s) P_{W|S}(w|s)^\alpha Q_W(w)^{1-\alpha} \\
&= \sum_{s} Q_W(w^*(s))^{1-\alpha} P_S(s) \\
&\le Q_{\min}^{1-\alpha},
\end{align*}
where $P_{W|S}(w')=\mathbf{1}_{\{w'=w^*\}}$ is a distribution that concentrates its mass on the hypothesis $w^*$ is defined as
\[w^*\in\text{argmax}_{w\in \mathcal{W}} \mathrm{KL}(\hat{L}(S,w)||L(w)),\] i.e. $w^*$ is any maximizer of the KL divergence between empirical and population loss.
\noindent
In this case the bound becomes
\begin{align*}
{P}\{E\} &= \mathbb{P}_S\left\{ 
\sup_{w} \mathrm{KL}(\hat{L}(S, w)||L(w)) 
\ge \frac{\log \frac{1}{\delta}}{n} 
\right\}\\
&\le \delta^{\frac{\alpha - 1}{\alpha}}
Q_{\min}^{\frac{1 - \alpha}{\alpha}}.
\end{align*}
By reparameterizing $\delta'$ as $\delta^{\frac{\alpha - 1}{\alpha}} 
Q_{\min}^{\frac{1 - \alpha}{\alpha}}$, we can then achieve
\begin{align*}
&\mathbb{P}_S \left\{
\exists w,\, \mathrm{KL}(\hat{L}(S, w)||L(w)) \geq 
\frac{
\log \frac{1}{Q_{min}} + \frac{\alpha}{\alpha - 1} \log \frac{1}{\delta'}
}{n}
\right\}\\
&\quad\quad\quad\quad\quad\quad\quad\quad\quad\quad\quad\quad\quad\quad\quad\quad\quad\quad\quad\quad \leq \delta',    
\end{align*}
or equivalently
\begin{align*}
&\mathbb{P}_S \left\{
\forall w,\, \mathrm{KL}(\hat{L}(S, w)||L(w)) \leq 
\frac{
\log \frac{1}{Q_{min}} + \frac{\alpha}{\alpha - 1} \log \frac{1}{\delta'}
}{n}
\right\}\\
&\quad\quad\quad\quad\quad\quad\quad\quad\quad\quad\quad\quad\quad\quad\quad\quad\quad\quad\geq 1-\delta'.    
\end{align*}
\end{proof}
The main novelty of our framework comes from specifying an "undesirable event" $E$, where the flexible choice of $E$ provides the flexibility for our framework but also achieves a tighter generalization bound. Therefore, defining an optimal and measurable "undesirable event" $E$ can be an interesting question to study in the future.
\begin{theorem}
\label{proof of theorem5}
    ($\mathcal{H}^p$-PAC-Bayes bound). Let $Q$ be a distribution over a finite hypothesis space $\mathcal{W}$ such that $Q_{\min} := \min_w Q_{W}(w)>0$. For any $p>1$, and for sufficiently small $\delta>0$ , with probability at least $1-\delta$ over $P_S$, it holds that
\begin{align}
\forall w,\ 
\mathrm{KL}(\hat{L}(S, w)||L(w))
\leq 
\frac{
\log\left[(Q_{min})^{1-p}\delta^{-p}{-1}\right] }{(p-1)n}.
\end{align}
\label{PAC-Bayes_with_H_p}

\end{theorem}
\begin{proof}
See Appendix \ref{Proof of Theorem5_appendix}
\end{proof}
\begin{theorem}
\label{chi-PAC-Bayes bound}
    ($\mathcal{\chi}^2$-PAC-Bayes bound). Let $Q$ be a distribution over a finite hypothesis space $\mathcal{W}$ such that $Q_{\min} := \min_w Q_{W}(w)>0$. For any $\delta\in(0,1]$, with probability at least $1-\delta$ over $P_S$, it holds that
\begin{align}
\forall w,\ 
\mathrm{KL}(\hat{L}(S, w)||L(w))
&\leq \frac{\log\frac{1}{Q_{min}}+2\log\frac{1}{\delta}}{n}.
\end{align}
\label{PAC-Bayes_with_chi_square}

\end{theorem}
\begin{proof}
See Appendix~\ref{Theorem5}
\end{proof}
\begin{remark}
 The DPI‑PAC‑Bayesian framework can be applied to an arbitrary f-divergence and yields generalization bounds whose relative tightness is governed by their divergence parameters. The $\mathcal{\chi}^2$-PAC-Bayes bound is parameter‑free, while both $D_\alpha$-PAC-Bayes and $\mathcal{H}^p$-PAC-Bayes bounds possess a free parameter---$\alpha>1$ and $p>1$---that modulates the trade‑off between the divergence penalty and the confidence term $\log(\frac{1}{\delta})$. 
\end{remark}
\subsection{Empirical Evaluation of Bounds}
We apply our bounds in a logistic classification problem in a 2-dimensional space, where $w\in\mathbb{R}^2$, $Z_i=(\textbf{x}_i,y_i)\in\mathbb{R}^3$. Each $\textbf{x}_i=\{(x_{i1},x_{i2})\}$ is sampled from a multivariate Gaussian distribution $\mathcal{N}(0,\textbf{I}_2)$. The label $y\in\{0,1\}$ is generated from the Bernoulli distribution with probability $p(y=1|\textbf{x}_i,w^*)=\frac{1}{1+e^{-\textbf{x}^T_i w^*}}$, where $w^*=(0.5,0.5)$. The generalization gap is measured by $\mathrm{KL}(\hat{L}(S,W)||L(W))$, where the loss function is given by 0-1 loss $\ell(Z_i,w)=I(\frac{1}{2}(\operatorname{sign}(x^T_i w)+1)\neq y_i)$.
We work with a finite hypothesis space $\mathcal{W}$ with $|\mathcal{W}| = 50$.  Each hypothesis is a weight vector $w \in \mathbb{R}^2$ whose coordinates are sampled independently from the uniform distribution $\mathrm{Unif}([-100,100])$.  
Because there is no prior information about the data, it is natural for us to assign the same importance (or probability) to each hypothesis, then we adopt the uniform prior distribution on $\mathcal{W}$ (i.e. ${Q_{min}}={Q_{W}(w)}$).

In Figure~\ref{fig:example}, we compare the tightness of the bounds derived by our framework, where we change the size of the training sample from 100 to 1600. For the ${D}_\alpha$-PAC-Bayes bound and the $\mathcal{H}^p$-PAC-Bayes bound, we experiment with different parameters, where $\alpha,\ p\in\{10,10^3,10^7\}$. To make a full comparison, we also compute the PAC-Bayes bound when the posterior is constrained to a point mass.
\begin{figure}[htbp]
  \centering
  \includegraphics[width=1\linewidth]{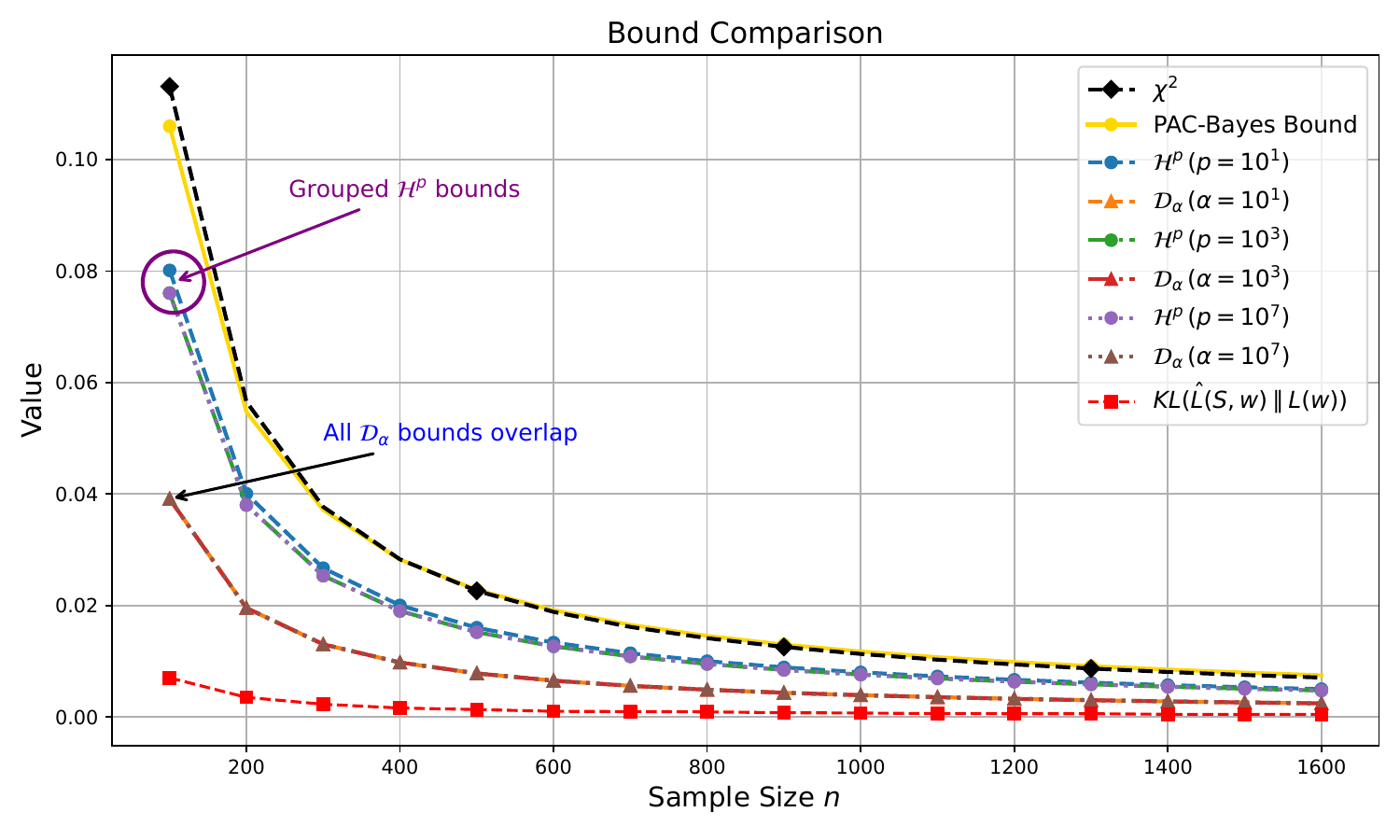}
  \caption{Comparison for the tightness of three bounds. $\delta=0.025$}
  \label{fig:example}
\end{figure}

Across the entire range of \(n\), the observed ordering of tightness is $
\mathcal{D}_{\alpha}\text{-PAC‑Bayes}<\mathcal{H}^{p}\text{-PAC‑Bayes}<
\text{PAC‑Bayes}.
$
While the $\chi^{2}$-PAC‑Bayes curve is the loosest among the variants, its gap to the standard PAC‑Bayes bound narrows as $n$ grows, making the two essentially comparable for large sample sizes. In summary, under the present experimental setting, the \(\mathcal{D}_{\alpha}\)-PAC‑Bayes bound delivers the most parameter‑robust capability and tightest guarantee.
\subsection{Connection to the OR bound and the PAC-Bayes bound}
The bounds developed by the DPI-PAC-Bayesian framework exhibit a close connection to the OR bound and the PAC-Bayes bound.  In particular, we will show in the sequel that when the prior distribution $Q_{W}(w)$ is chosen to be uniform, the $D_\infty$-PAC-Bayes and $\mathcal{H}^\infty$-PAC-Bayes bounds recover to the OR bound. Additionally, our bounds are in spirit similar to the PAC-Bayes bound, but our bounds are provably tighter than the PAC-Bayes bound in the special case. 
An important note is that the bounds converge monotonically to their tightest forms when $\alpha, \ p\rightarrow\infty$.

\begin{corollary}[Limiting $D_{\infty}$/ $\mathcal{H}^{\infty}$-PAC‑Bayes bound]
\label{cor:optimal}
Let the prior $Q_{{W}}$ be uniform over $\mathcal{W}$. For any $\delta\in(0,1]$, with probability at least $1-\delta$ over $P_S$, we have  
\begin{equation*}
\forall\,w,\ 
\mathrm{KL}\!(\hat{L}(S,w)||L(w))\;\le\;
\frac{ \log\frac1{Q_{{W}}(w)}\;+\;\log\frac1\delta}{n}.
\end{equation*}
The same bound is obtained in either of the following limits: $\alpha\to\infty$ in the $D_\alpha$‑PAC‑Bayes bound; $p\to\infty$ in the $\mathcal{H}^p$‑PAC‑Bayes bound.

\end{corollary}  
\begin{proof}
    See Appendix~\ref{Proof_corollary}
\end{proof}
\begin{corollary} 
\label{PAC-Bayeswithchi_infty}
(\textit{$\mathcal{\chi}^2$}-PAC-Bayes bound). Let the prior $Q_{{W}}$ be uniform over $\mathcal{W}$. For any $\delta\in(0,1]$, with probability at least $1-\delta$ over $P_S$, it holds that
\begin{equation*}
\forall w, \; \mathrm{KL}(\hat{L}(S, w)||L(w)) \leq \frac{ \log \frac{1}{Q_{W}(w)}+2\log \frac{1}{\delta} }{n}.
\end{equation*}
\end{corollary} 
Importantly, compared to the PAC-Bayes bound \eqref{pac-bayes-pointmass}, both the \(D_\infty\)-PAC‑Bayes bound and the \(\mathcal H^\infty\)-PAC‑Bayes bound remove the extra term $\log2\sqrt{n}$, and these two bounds recover the same expression of the OR bound in Theorem \ref{or_bound_theorem}.

Additionally, the $D_\infty$-PAC-Bayes bound and the $\mathcal{H}^\infty$-PAC-Bayes bound provide a tighter generalization guarantee than the \(\chi^{2}\)-PAC‑Bayes bound by a margin of \(\log(1/\delta)\!/n\).
\section{Acknowledgment}
We acknowledge the helpful discussion with Dr. Matthias Frey.



\newpage

\bibliographystyle{plain}
\bibliography{references} 

%
\newpage
\appendix

\subsection{Proof of Lemma~\ref{DPI_original_bound_Chi_Squared_lemma}}
\label{proof_lemma}
For Pearson $\chi^2$-divergence, we have the following exact relationship:
\[
\chi^2(P||Q)=e^{D_2(P||Q)}-1,
\]
where $D_2(P||Q)$ is the Rényi divergence with $\alpha=2$.

Equivalently,
\begin{align}
\label{relationship}
e^{\frac{1}{2}D_2(P||Q)}=(\chi^2(P||Q)+1)^{\frac{1}{2}},
\end{align}

Hence, Lemma~\ref{DPI_original_bound_Chi_Squared_lemma} follows directly by substituting $\alpha=2$ into Lemma~\ref{DPI_original_bound_1} and applying \eqref{relationship}. 

\subsection{Proof of Theorem~\ref{proof of theorem5}}
\label{Proof of Theorem5_appendix}
We adopt the same joint distributions \( Q_{S, W} \) and \( P_{S, W} \) as Theorem~\ref{PAC-Bayes_with_D_alpha}, and define the events $E$ and $E_w$ as in equations~\eqref{Event_general} and~\eqref{Event_for_specific_w}, respectively. 
Let the posterior be defined as \( P_{W|S}(w')=\mathbf{1}_{\{w'=w^*\}}\), where $w^*\in\text{argmax}_{w\in \mathcal{W}} \mathrm{KL}(\hat{L}(S,w)||L(w))$. 

We next show that for sufficiently small $\delta$, we have $P(E),Q(E) \leq 1/2$.
Now we have
$${P}\{E\} = \mathbb{P}_S\left\{ 
\sup_{w} \mathrm{KL}(\hat{L}(S, w)||L(w)) 
\ge \frac{\log \frac{1}{\delta}}{n} 
\right\}.$$

We consider a random variable
$$
K\;:=\;\sup_{w}\mathrm{KL}\bigl(\hat L(S,w)\,\Vert\,L(w)\bigr) 
$$
and its cumulative distribution function (CDF) $F_K(t)\;:=\;\mathbb P_S\{K\le t\},\ t\in\mathbb{R}$.

We also define 
$$
G(\delta)\;:=\;\mathbb P_S\!\Bigl\{\,K\;\ge\;\tfrac{\log(1/\delta)}{n}\Bigr\}
\;=\;
1-F_K\!\Bigl(\tfrac{\log(1/\delta)}{n}\Bigr),
$$
where $G(\delta)={P}\{E\}.$

For any $t\ge0$, $0\le F_K(t)\le1$ and also $F_K(t)$ is non-decreasing, hence $1-F_K(t)$ is non-increasing and also $\lim\limits_{t\to\infty} 1-F_K(t)=0$.
Therefore, $G(\delta)$ is non-increasing and $\lim\limits_{\delta\to0} 1-F_K(\frac{\log(1/\delta)}{n})=0$, or equivalently $\lim\limits_{\delta\to0}G(\delta)=0$.

Thus we can pick $\delta_0\in (0,1]$ such that $G(\delta_0)<\frac{1}{2}$, and every $\delta\le\delta_0$ satisfies 
$$G(\delta)\le G(\delta_0)<\frac{1}{2}.$$

Therefore, choosing $\delta$ sufficiently small guarantees that $G(\delta)=P\{E\}<\frac{1}{2}$.

By the same argument (with \(P\) replaced by \(Q\)), we also have \(\displaystyle\lim_{\delta\to0} Q(E)=0\), and hence in particular $P(E),Q(E) < \frac{1}{2}$ for sufficiently small $\delta$.  This allows us to apply Lemma~\ref{DPI_original_bound_hellinger_p_lemma} in the proof of Theorem~\ref{proof of theorem5}.

We have the inequality \eqref{DPI_original_bound_hellinger_p_inequality} in Lemma~\ref{DPI_original_bound_hellinger_p_lemma}, where two terms---$\left[1+Q(E)^{1-p}\right]^{-\frac{1}{p}}$ and $\mathcal{H}^P(P(S,W)||Q(S,W))$---need to be upper bounded.

When $p\ge1$, finding the upper bound for the term $\left[1+Q(E)^{1-p}\right]^{-\frac{1}{p}}$ is equivalent to finding the upper bound of $Q(E)$. We still use the test-set bound $Q(E)=\sum\limits_wQ_W(w)\mathbb{P}(E_w)\le\delta\sum\limits_wQ_W(w)=\delta$. Also, we have 
\begin{align*}
\mathcal{H}^P(P(S,W)||Q(S,W))
&=\frac{\sum\limits_{w,s}P(s,w)^pQ(s,w)^{1-p}-1}{p-1}\\
&=\frac{\sum\limits_{s}P_S(s)Q_W(w^*(s))^{1-p}-1}{p-1}\\
&\leq\frac{\left[\left(Q^{1-p}_{min}\sum\limits_{s}P_S(s)\right)-1\right]}{p-1}\\
&\leq\frac{\left(Q^{1-p}_{min}-1\right)}{p-1}.  
\end{align*}
Applying the above inequalities to  Lemma~\ref{DPI_original_bound_hellinger_p_lemma}, we can achieve 
\[
P(E)\leq \left(Q_{min}\right)^\frac{1-p}{p}(1+\delta^{1-p})^{-\frac{1}{p}}.
\]
Thus we get
\begin{align*}
{P}\{E\} = &\mathbb{P}_S\left\{ 
\sup_{w} \mathrm{KL}(\hat{L}(S, w)||L(w)) 
\ge \frac{\log \frac{1}{\delta}}{n} 
\right\}\\ 
&\le \left(Q_{min}\right)^\frac{1-p}{p}(1+\delta^{1-p})^{-\frac{1}{p}}.    
\end{align*}

By reparameterizing $\delta'$ as $\left(Q_{min}\right)^\frac{1-p}{p}(1+\delta^{1-p})^{-\frac{1}{p}}$, we can then achieve
\begin{align*}
&\mathbb{P}_S \left\{\exists w,\, \mathrm{KL}(\hat{L}(S, w)||L(w)) \geq 
\frac{\log \left[\frac{1}{\delta^{p}}(Q_{min})^{1-p}-1\right]}{n(p-1)}\right\}\\
&\quad\quad\quad\quad\quad\quad\quad\quad\quad\quad\quad\quad\quad\quad\quad\quad\quad\quad\leq \delta',    
\end{align*}
or equivalently
\begin{align*}
 &\mathbb{P}_S \left\{\forall w,\, \mathrm{KL}( \hat{L}(S, w)||L(w)) \leq 
\frac{\log \left[\frac{1}{\delta^{p}}(Q_{min})^{1-p}-1\right]}{n(p-1)}\right\}\\
&\quad\quad\quad\quad\quad\quad\quad\quad\quad\quad\quad\quad\quad\quad\quad\quad\quad\quad\geq1- \delta'. 
\end{align*}
\subsection{Proof of Theorem~\ref{chi-PAC-Bayes bound}}
\label{Theorem5}
Following the proof procedures in Theorem~\ref{PAC-Bayes_with_D_alpha} and Theorem~\ref{PAC-Bayes_with_H_p}. Lemma~\ref{DPI_original_bound_Chi_Squared_lemma} is applied, we bound $Q(E)^{\frac{1}{2}}$ by KL test-set bound
\[Q(E)^{\frac{1}{2}}\le (\sum\limits_w Q_W(w)\mathbb{P}\{E_w\})^{\frac{1}{2}}=\delta^{\frac{1}{2}}.\]
Also we have\begin{align*}
\left(\chi^2\left(P(s,w)||Q(s,w)\right)+2\right)^\frac{1}{2}&=\left(\sum\limits_{s,w}\frac{P(s,w)^2}{Q(s,w)}-1+2\right)^\frac{1}{2}\\
&\le\left(\frac{\sum\limits_{s}P_S(s)}{Q_{min}}+1\right)^\frac{1}{2}\\
&=\left(\frac{1+Q_{min}}{Q_{min}}\right)^\frac{1}{2}.
\end{align*}
Then we can achieve the bound in Theorem~\ref{PAC-Bayes_with_chi_square} by the reparameterization used in both Theorem~\ref{PAC-Bayes_with_D_alpha} and Theorem~\ref{PAC-Bayes_with_H_p}.
\subsection{Proof of Corollary~\ref{cor:optimal}}
\label{Proof_corollary}
Corollary~\ref{cor:optimal} contains two bounds.  
The RHS of the $\mathcal D_\alpha$–PAC‑Bayes bound contains the multiplicative factor
\(\tfrac{\alpha}{\alpha-1}\).
Since
\[
\lim_{\alpha\to\infty}\frac{\alpha}{\alpha-1}=1,
\]
the entire expression converges to the same limit as \(\alpha\to\infty\).
Hence its limiting form is immediate and a separate proof is omitted.

We focus below on the $\mathcal H^{p}$–PAC‑Bayes bound.

The RHS of Theorem~\ref{PAC-Bayes_with_H_p} is $$\frac{\log\!\bigl[(Q_{\min})^{1-p}\,\delta^{-p}-1\bigr]}{(p-1)\,n}.$$ 
As $p\to \infty$, we have:
\begin{align*}
\lim\limits_{p\to\infty}
\frac{\log\!\bigl[(Q_{\min})^{1-p}\,\delta^{-p}-1\bigr]}{(p-1)\,n}
=\lim_{p\to\infty}
\frac{\log\!\Bigl[\bigl(\tfrac{1}{\delta\,Q_{\min}}\bigr)^{p}-1\Bigr]}{p\,n}.
\end{align*}
We have $\lim\limits_{p\to\infty}(\frac{1}{\delta\,Q_{\min}})^p > 1$ because $0 < \delta \le 1$ and $ 0 < Q_{\min} < 1$, and also we have $\lim\limits_{p\to\infty} p\,n = \infty $.
Therefore, by applying L'Hôpital's rule, we can get the limiting form of the RHS of the $\mathcal{H}^p$-PAC-Bayes bound 
$$\lim_{p\to\infty}
\frac{\log\!\Bigl[\bigl(\tfrac{1}{\delta\,Q_{\min}}\bigr)^{p}-1\Bigr]}{p\,n}
=\frac{\log(1/Q_{\min}) + \log(1/\delta)}{n}.$$

By setting the prior distribution as uniform, we can then replace $Q_{min}$ with $Q_{W}(w)$ and achieve the same form in Corollary~\ref{cor:optimal}.

\end{document}